\documentclass[a4paper,12pt]{article}
\usepackage[a4paper, total={7in, 8in}]{geometry}

\RequirePackage[OT1]{fontenc}
\RequirePackage{amsthm,amsmath}

\usepackage[utf8]{inputenc} 
\usepackage[T1]{fontenc}    
\usepackage{hyperref}
\usepackage{url}            
\usepackage{booktabs}       
\usepackage{amsfonts}       
\usepackage{nicefrac}       
\usepackage{microtype}      
\usepackage{lipsum}
\usepackage[english]{babel}

\usepackage{amssymb}
\usepackage{amsmath}
\usepackage{comment}
\usepackage{amsthm}
\usepackage{bbm}
\usepackage{mathabx}
\usepackage{verbatim}
\usepackage{enumitem}
\usepackage{mathabx}
\usepackage{graphicx}

\usepackage{natbib}
\theoremstyle{plain}

\newtheorem{theorem}{Theorem}[section]
\newtheorem{lemma}[theorem]{Lemma}

\newtheorem{algorithm}{Algorithm}
\newtheorem{assumption}{Assumption}
\newtheorem{remark}{Remark}
\newtheorem{prop}{Proposition}[section]
\theoremstyle{definition}

\newcommand{\norm}[1]{\left\lVert#1\right\rVert}
\newcommand{\tildehat}[1]{\tilde{\widehat{#1}}}

\newcommand{\duohat}[1]{\widehat{\widehat{#1}}}

\begin{document}
\bibliographystyle{plain}
\title{Model-free Bootstrap Prediction Regions for Multivariate Time Series}
\author{Yiren Wang\footnote{Department of Mathematics, 
		Univ.~of California, San Diego;  email: yiw518@ucsd.edu}
	\and Dimitris N. Politis\footnote{Department of Mathematics and Halicio\v{g}lu Data Science Institute, Univ.~of California, San Diego;  email: dpolitis@ucsd.edu}}
\date{}

\maketitle
\begin{abstract}
In \cite{Das2017}, a model-free bootstrap(MFB) paradigm was proposed for generating prediction intervals of univariate,  (locally) stationary time series. Theoretical guarantees for this algorithm was resolved in \cite{WangPolitis2019}  under stationarity and weak dependence condition. 
Following this line of work, here we extend MFB for predictive inference under a multivariate time series setup. 
We describe two algorithms, the first one works for a particular class of time series under any fixed dimension $d$; the second one works for a more generalized class of time series under low-dimensional setting.
We justify our procedure through theoretical validity and simulation performance.
\end{abstract}\maketitle

\section{Introduction}\label{sec:multivariate_intro}
Time series forecasting(prediction) is widely applicable in many different fields where the prediction is carried out by conditioning on previous observations. 
Traditional approaches for time series prediction often rely on a parametric model assumption that can describe the explicit dependence relations of the data. 
The space of viable time series models that have been studied in the past is huge, therefore the search of good models that can describe the data requires both good understanding of domain knowledge as well as statistical expertise in order to execute model fitting and goodness-of-fit tests. 

In a recent monograph,  \cite{Politis2015} proposed the idea of {\it model-free} prediction in the context of regression and time series problems. 
The idea behind it is to find a one-to-one transform $H$ that takes the original data $\{X_t\}_{t=1}^n$ into a new sequence of data that are i.i.d.(independent and identically distributed), then the (inverse) transform along with resampling in  the i.i.d. world can be combined to construct both valid point predictors and prediction intervals. 
Following this, \cite{Das2017} proposed a model-free bootstrap(MFB) algorithm for generating prediction intervals for locally stationary time series that outperforms model-based approach. 
This algorithm was further studied in \cite{WangPolitis2019} under stationary weakly-dependent setting, where bootstrap validity was proved for both prediction interval and confidence intervals for linear statistics and spectral density. 
In this paper, we extend the model-free bootstrap algorithm of \cite{WangPolitis2019} for prediction regions of multivariate time series. We demonstrate the effectiveness of the algorithm by both showing theoretical bootstrap validity and also numerical simulations.

\section{Description of algorithm}\label{sec:multivariate_algorithm}
\subsection{Models of interest}
Let $\underline Y_t = (Y_{1,t},\cdots, Y_{d,t})^{\mathsf T}$ represent a $d$-dimensional time series. We consider $\underline Y_t$ to be from one of the following models:

\textbf{Model 1:} Let $\{f_i: \mathbb R \rightarrow \mathbb R\}_{i=1}^n$ be strictly monotone, continuous functions. $\underline W_t\in\mathbb R^d$ is a purely non-deterministic, strictly stationary Gaussian process. $\underline Y_t$ satisfies
\begin{equation}\label{model:1}
		Y_{i,t} = f_i(W_{i,t}),\quad i = 1,\cdots d.
\end{equation}
Model \ref{model:1} is nothing more than a multivariate analog of the class of time series models studied in \cite{WangPolitis2019}. Interestingly, it also has connection  with copula models for multivariate time series, which we will investigate later. Let us introduce a second class of model:

\textbf{Model 2:} Let $f_i(\cdot; \underline Y_{1:i-1,t}): \mathbb R\rightarrow \mathbb R$ be strictly monotone(increasing) functions whose parameters also depend on the previous $i-1$ dimensional entries of $\underline Y_t$: $\underline Y_{1:i-1,t}$, with the exception of $i=1$, where $f_1$ is a deterministic function. The time series $\underline Y_t$ are defined sequentially in dimensional order in the following manner:

\begin{equation}\label{model:2}
	\begin{aligned}
		& Y_{1,t} = f_1(W_{1,t});\\
		& Y_{i,t} = f_i(W_{i,t}; \underline Y_{1:i-1,t}), i\geq 2.
	\end{aligned}
\end{equation}
Model 2 is  more complicated than model 1 in that we allow for the parameters of subsequent transfer functions $f_i, i\geq 2$ to be dependent on previous entries of $\underline Y_t$. Model 1 can be recovered by setting $f_i(\cdot; \underline Y_{1:i-1,t}) = f_i$.

Because of the monotonicity of transfer functions,  the map $\underline W_t \rightarrow \underline Y_t$ is invertible for both models, thus
the information set generated by $\underline Y_t$ is equivalent to the one generated by $\underline W_t$, i.e., $$\mathcal F_t = \sigma(\underline Y_s, s\leq t) = \sigma(\underline W_s, s\leq t).$$

\subsection{Connection to  copula-based models}\label{sec:multivariate_copularelationship}
Another active line of research for forecasting multivariate time series involves utilizing a copula representation. By Skyler's theorem, an arbitrary time series $\underline Y_t$ can be fully described by the conditional distributions of $Y_{i,t}|\mathcal F_{t-1} \sim F_i(|\mathcal F_{t-1})$, and a conditional copula function $ C(\boldsymbol\cdot|\mathcal F_{t-1}):[0,1]^d \rightarrow \mathbb  R$ such that 
$\forall \underline y = (y_1,\cdots,y_d)^{\mathsf T}$,
\begin{equation*}
  F(\underline y|\mathcal F_{t-1}) =  C(F_1(y_1|\mathcal F_{t-1}),\cdots,F_d(y_d|\mathcal F_{t-1})|\mathcal F_{t-1} ).
\end{equation*}
The copula representation offers an approach to decorrelate spatial dependence and serial dependence. Note that in the econometrics literature, the conditional distribution of each dimension are often assumed to be fixed, i.e., $F_i(|\mathcal F_{t-1}) = F_i$, to guarantee effectiveness of valid statistical procedures, see \cite{PATTON2013899}.
One of the most famous copulas is the Gaussian copula: a centered $d-$dimensional  Gaussian random vector $\underline Z $ with correlation matrix $\boldsymbol \Sigma$ has distribution function
\begin{equation}\label{eq:copula}
\Phi_{\boldsymbol\Sigma,d}(\underline z) =  C_{\boldsymbol \Sigma}(\underline u),
\end{equation}
where $\underline u = (u_1, \cdots, u_d)^{\mathsf T}$with $u_i = \Phi(z_i)$; and $ C_{\boldsymbol\Sigma}(\underline u) = \Phi_{\boldsymbol \Sigma,d}(\Phi^{-1}(u_1),\cdots,\Phi^{-1}(u_d))$ is the Gaussian copula with correlation matrix $\boldsymbol \Sigma$.
Generalizing this to the case of a stationary Gaussian processes denoted by $\underline Z_t$, since the conditional distribution of $\underline Z_t|\mathcal F_{t-1}$ is multivariate normal, it can be fully represented by the Gaussian copula \eqref{eq:copula} above, where $\boldsymbol\Sigma = \boldsymbol \Sigma_{t|t-1}$ will be the correlation matrix of the conditional normal distribution. Thus $ C(|\mathcal F_{t-1}) =  C_{\boldsymbol{\Sigma}_{t|t-1}}$.

Another well known fact is that the copula of a joint distribution stays unchanged under monotone function transformations of each marignal random variable. 
To put it under the setup of model \eqref{model:1} and assume that $F_i(|\mathcal F_{t-1}) = F_i$, $$ F(\underline y|\mathcal F_{t-1}) =  C_{\boldsymbol \Sigma}(F_1(y_1),\cdots,F_d(y_d) ),$$ where $ C_{\boldsymbol \Sigma}$ is the conditional Gaussian copula of $\underline W_t|\mathcal F_{t-1}$.
Thus, the time series $\underline Y_t$ of model \ref{model:1} inherits the same conditional copula as the Gaussian process $\underline W_t$. $\underline Y_t$ is also referred to as the Gaussian copula process, see \cite{NIPS2010_fc8001f8}.
However, this does not apply to the  case of model \ref{model:2} where the transfer function parameters depends on other indices of $\underline Y_t$. Therefore, model \ref{model:2} can be more complicated in that the conditional copula function can go beyond the Gaussian copula.

Notably, the model-free bootstrap to be introduced in the next section shares the same spirit of the semi-parametric method (cf \cite{PATTON2013899} for more details) for copula process prediction. To elaborate, the model-free bootstrap under model \eqref{model:1} requires both estimation of the marginal CDFs via nonparametric method, as well as consistent estimation for the autocovariance structure of the underlying Gaussian process. Similarly, the semi-parametric approach in \cite{PATTON2013899} assumes general distributions for the CDF which are estimated via nonparametric methods, while the copula function is assumed to belong to a particular parametric family(e.g, Gaussian), and then estimation for the copula is carried out parametrically. 
\subsection{The model-free bootstrap algorithm}
The model-free bootstrap of \cite{Politis2015} offers a general principle for conducting bootstrap resampling procedure for parameter inference and predictive inference under regression or time series setups. 
In the univariate time series setting, the algorithm relies on an invertible transform $H_n$ between a length $n$ time series data $\{X_t\}_{t=1}^n$, and a set of $n$ i.i.d. random variables $\{\xi_t\}_{t=1}^n$, where $H_n$ is sample size adaptive. 
For parameter inference, let $\theta$ be the parameter of interest which is estimated by some statistic $\widehat\theta_n$. 
The model-free bootstrap is first performed in the space of i.i.d. random variables to generate $\{\xi_t^{*}\}_{t=1}^n$, and then we use the $H_n^{-1}$ to get bootstrap samples $\{X_t^{*}\}_{t=1}^n$. 
Then the distribution of $\widehat\theta_n - \theta$ can be approximated by $\widehat\theta_n^* - \widehat\theta_n$, where $\widehat\theta_n^*$ is calculated based on the bootstrap sample.

Predictive inference is in a more difficult situation, as valid predictions should be carried out conditioning on the entire  observed series $\{X_t\}_{t=1}^n$.
First of all, we need a notion called the predictive root, an analog to the pivot random variable $\widehat\theta_n - \theta$ in previous setting, which is defined as $$ r_{n+1} = X_{n+1} - \widehat X_{n+1},$$
where $X_{n+1}\sim F_{n+1|n}$ is the next unobserved data, and $\widehat X_{n+1}$ is a sample point predictor for the future observation, 
such that 
$$\widehat X_{n+1} = \arg\min_{x\in\mathbb R} \mathbb E_{X_{n+1}\sim \widehat F_{n+1|n}}\mathcal L(x,X_{n+1}).$$
Here $\widehat F_{n+1|n}$ is the estimated conditional CDF based on the estimated transform $\widehat H_n$. 
Examples for the loss function $\mathcal L$ include the $L^1$ and $L^2$
loss minimizers, for which we have $L^1$/$L^2$-optimal predictors, respectively. Essentially, $\widehat X_{n+1}$ is the sample estimator for the actual $1$-step ahead predictor:
$$\mathcal P(X_{n+1}|\mathcal F_n) =  \arg\min_{x\in\mathbb R} \mathbb E_{X_{n+1}\sim  F_{n+1|n}}\mathcal L(x,X_{n+1}).$$

By approximating the conditional distribution of $r_{n+1}$ via bootstrap, we can construct a two-sided  prediction interval of size $1-\alpha$ around the point predictor as follows: $$\left( \widehat X_{n+1} +  L_{\alpha/2}^*,\widehat X_{n+1} + R_{\alpha/2}^*\right),$$
where $ L_{\alpha/2}^*$ and $ R_{\alpha/2}^*$ are the lower/higher $\alpha/2$-quantiles for the distribution of $r^*_{n+1}$, the bootstrap version of $r_{n+1}$.

The foremost problem is how to construct the invertible transform $H_n$ towards i.i.d.-ness. As \cite{Politis2015} points out, such an invertible transform always exists. However, $H_n$ needs to be estimated based on data, and further structural assumptions both simplifies the estimation procedure and also guarantees certain level of efficiency required for the bootstrap to be valid. \cite{Das2017} proposed the following $H_n$ based on the probability integral transform(PIT) under the structural assumption that $X_t$ is a monotone transform of a stationary Gaussian process: 
\begin{enumerate}
	\item Let $F_X$ be the CDF of $X_t$; Let $U_t = F_X(X_t)$.
	\item Let $Z_t = \Phi^{-1}(U_t)$. \cite{Das2017} showed that $Z_t$ is a Gaussian process. Let $\underline Z_n = (Z_1, \cdots, Z_n)^{\mathsf T}$
	\item Let $\boldsymbol\Xi_n$ be the $n$-dimensional autocovariance matrix of $Z_t$, and $\boldsymbol\Xi_n^{1/2}$ the upper Cholesky decomposition matrix of $\boldsymbol\Xi$. Then $\underline \xi_n = \boldsymbol\Xi_n^{-1/2}\underline Z_n$ consists of $n$ i.i.d. standard normal random variables.
\end{enumerate}
$H_n^{-1}$ will map $n$ i.i.d. normal random variables $\underline \xi_n$ back to $\{X_t\}_{t=1}^n$ in the following way:
\begin{enumerate}
	\item Let $\underline Z_n = \boldsymbol\Xi_n^{1/2} \underline \xi_n$;
	\item let $U_t = \Phi(Z_t)$, $X_t = F_X^{-1}(U_t)$.
\end{enumerate}

\cite{WangPolitis2019} showed that under certain weakly dependence assumption of $X_t$, one can efficiently estimate the CDF $F_X$ with a nonparametric estimator. By using an augmented version of $\Phi$ in the transform, $\boldsymbol\Xi_n$ and its inverse can also be efficiently estimated, thereby showing consistency of the estimated transforms $\widehat H_n$ and $\widehat H_{n}^{-1}$, and also validity of the model-free bootstrap.

To create the predictive root $r_{n+1}^*$, both $X_{n+1}$ and $\widehat X_{n+1}$ need to be resampled separately. To sample $\widehat X_{n+1}$, with $\underline \xi_n^*$ being sampled from $\underline \xi_n$, we can use the above $
\widehat H_n^{-1}$ to get the bootstrap samples $\underline X_n^*$, which is used to re-estimate $\widehat H_n$ to get $\widehat F^*_{n+1|n}$. Then 
$$\widehat X_{n+1}^* = \arg\min_{x\in\mathbb R} \mathbb E_{X_{+1}\sim \widehat F^*_{n+1|n}}\mathcal L(x,X_{n+1}).$$
As for $X_{n+1}^*$, we first extend $\widehat H_n^{-1}$ to $\widehat H_{n+1}^{-1}$, which is then used to map the vector $(\underline\xi_n,\xi_{n+1}^*)$ back to $(\underline X_n, X_{n+1}^*)$. It is easy to see that by doing this, $X_{n+1}^*\sim \widehat F_{n+1|n}$.

We next extend the MFB to multivariate time series under model \ref{model:1} or \ref{model:2} following a similar route as above. Specifically, we first transform $\underline Y_t$ to a centered multivariate Gaussian process $\underline Z_t$ using the PIT, and then further whiten it to get i.i.d. normal vectors by  decorrelating $\underline Z_t$ with its covariance structure. The inverse of the above transforms are used to construct the transform bootstrap samples from the i.i.d. space back to the space of time series.

In order to whiten the Gaussian process $\underline Z_t$ which has $n\times d$ observations, we need the following technique adapted from \cite{Jentsch}. 
The entire sequence $\mathbf Z = [\underline Z_1,\cdots, \underline Z_n]$ can be flattened by stacking the observations into one row vector: $\underline Z_{dn} = vec(\mathbf Z) = [\underline Z_1^{\mathsf T},\cdots,\underline Z_n^{\mathsf T}]$. $\underline Z_{dn}$ is a multivariate normal vector, whose covariance matrix is symmetric block Toeplitz, and has the following form:

\begin{equation}
	\boldsymbol \Gamma_{dn} = 
	\begin{bmatrix}
		\boldsymbol \Gamma_{0} & \boldsymbol \Gamma_{1} &\cdots &\boldsymbol \Gamma_{n-1}\\
		\boldsymbol \Gamma_{1}^{\mathsf T} & \boldsymbol \Gamma_{0} &\cdots &\boldsymbol \Gamma_{n-2}\\
		 \vdots & \ddots &\ddots&\vdots\\
		\boldsymbol \Gamma_{n-1}^{\mathsf T} & \boldsymbol \Gamma_{n-2}^{\mathsf T} &\cdots &\boldsymbol \Gamma_{0}\\
	\end{bmatrix}
\end{equation}
where $\boldsymbol \Gamma_h = Cov(\underline Z_0, \underline Z_h) = \mathbb E \underline Z_{0} \underline Z_h^{\mathsf T}$ is the $d\times d$ lag-$h$ autocovariance matrix of $\underline Z_t$. 
Next, $\underline Z_{dn}$ can be whitened through left-multiplying $\boldsymbol \Gamma_{dn}^{-1/2}$, which results in i.i.d. standard normal variables.

To consistently estimate $\boldsymbol \Gamma_{dn}$, we require to use the flat-top estimator introduced in \cite{mcmurrypolitis2010}. 
The autocovariance matrix at lag $h>0$ can be estimated through the usual estimator
$$\widehat{\boldsymbol{\Gamma}}_h = \frac{1}{n} \sum_{t=1}^{n-h} \underline Z_t\underline Z_{t+h}^{\mathsf T}.$$
While for $h<0$ we can use $\widehat{\boldsymbol\Gamma}_h = \widehat{\boldsymbol\Gamma}_{|h|}^{\mathsf T}$.
Let $\kappa_l$ be the flat-top kernel with base function $\kappa$ and bandwidth parameter $l$, such that $\kappa_l(x) = \kappa(x/l)$. The new estimator is defined as 
\begin{equation}\label{eq:taperedestimator}
	\widehat{\boldsymbol\Gamma}_{\kappa, l} = \left( \kappa_l(i-j)\widehat{\boldsymbol \Gamma}_{|i-j|} \right)_{1\leq i,j\leq n},
\end{equation}
i.e., the autocovariance matrices at large lags are shrunk towards $\mathbf 0$. It is well known that with appropriate rate of divergence for $l$ as $n\rightarrow\infty$, the flat-top estimator is consistent to $\boldsymbol \Gamma_{dn}$ in operator norm:
$$\norm{\widehat{\boldsymbol\Gamma}_{\kappa, l}  - \boldsymbol\Gamma_{dn}}_{op}\overset{P}{\rightarrow} 0.$$
A common tapering function is the following trapezoid function:
\begin{equation}\label{eq:taper}
\kappa(x) = \begin{cases}
	1, & |x|\leq 1\\
	2 - |x|, & 1<|x|\leq 2\\
	0, & \text{otherwise}
\end{cases}.
\end{equation}

We also extend the concept of (1-step ahead) predictive root to the multivariate scenario. Let $\underline R_{n+1} = \underline Y_{n+1} - \underline{\widehat Y}_{n+1}$, where $\underline Y_{n+1}$ is the future observation conditioning on $\{\underline Y_t\}_{t=1}^n$; and $\underline{\widehat Y}_{n+1}$ is the $1$-step ahead predictor that satisfies 
\begin{equation}\label{eq:predictor}
	\widehat{\underline Y}_{n+1} = \arg\min_{\underline y\in\mathbb R^d} \mathbb E_{\underline Y_{n+1}\sim \widehat{F}_{n+1|n}}\mathcal L(\underline y, \underline Y_{n+1}),
\end{equation}
where $\widehat{F}_{n+1|n}$ is the estimated conditional CDF. Similarly, $\widehat{\underline Y}_{n+1}$ is an estimator for the true $1$-step ahead predictor 

$$\mathcal P(\underline Y_{n+1}|\mathcal F_n) = \arg\min_{\underline y\in\mathbb R^d} \mathbb E_{\underline Y_{n+1}\sim {F}_{n+1|n}}\mathcal L(\underline y, \underline Y_{n+1}).$$

\begin{algorithm}\label{alg:1}Bootstrap algorithm for $1-$step ahead prediction region under model  \ref{model:1}.

\begin{enumerate}
	\item For each dimension $i$, estimate the marginal CDFs $F_i$ via empirical or the nonparametric CDF estimator, denoted by $\widehat F_i$.
	\item Let $\widehat U_{i,t} = \widehat F_i(Y_{i,t})$; let $\tilde\Phi_c^{-1}$ be the quantile function of a thresholded normal distribution(see \cite{WangPolitis2019} for further details). Let $\tildehat{Z}_{i,t} = \tilde\Phi_c^{-1}(U_{i,t})$ which are estimations for the destination Gaussian process. Estimate the covariance structure of $\tildehat{\underline Z}_{dn}$, denoted by $\duohat{\boldsymbol\Gamma}_{dn}$, with the tapered covariance matrix estimator \eqref{eq:taperedestimator}. 
	\item  
	Flatten $\tildehat{\mathbf Z} = [\tildehat{\underline Z}_1,\cdots,\tildehat{\underline Z}_n]$ as $\tildehat{\underline Z}_{dn} = vec(\tildehat{\mathbf Z}) = [\tildehat{\underline Z}_1^{\mathsf T}:\cdots:\tildehat{\underline Z}_n^{\mathsf T} ]^{\mathsf T}$. Let $\widehat{\underline\xi}_{dn} = \duohat{\boldsymbol\Gamma}_{dn}^{-1/2} \tildehat{\underline Z}_{dn}.$ Also, based on the estimations in previous steps, calculate $\widehat{\underline Y}_{n+1} $ by equation \eqref{eq:predictor}.
	\item (bootstrap) 
	\begin{enumerate}[label=(\alph*)]
		\item Let ${\underline\xi}_{dn}^*$ be a vector of i.i.d. random variables uniformly sampled with replacement from the entries of $\widehat{\underline\xi}_{dn}$.
		Let ${\underline Z}_{dn}^* = \duohat{\boldsymbol\Gamma}_{dn}^{1/2}{\underline\xi}_{dn}^*$, based on which we get $\underline Z_t^*$, $t = 1,\cdots,n$. Then $Y_{i,t}^* = \widehat F_i^{-1}\left(\Phi(Z_{i,t}^*)\right)$. Use $\mathbf Y_n^* = [\underline Y_1^*,\cdots, \underline Y_n^*]$ to re-estimate the transforms above to get $\widehat F_{n+1|n}^*$, then get $\widehat{\underline Y}_{n+1}^*$ by equation  \eqref{eq:predictor}.
		\item Let $\underline\xi^*_d= (\xi_{dn+1}^*,\cdots,\xi_{d(n+1)}^*)$, with indices uniformly sampled from $\widehat{\underline\xi}_{dn}$, and $\underline\xi^*_{d(n+1)} = (\widehat{\underline\xi}_{dn}^{\mathsf T}, \underline\xi^*_d)^{\mathsf T}$.
		Then $\duohat{\boldsymbol \Gamma}_{d(n+1)}^{1/2}\underline\xi^*_{d(n+1)} : = (\tildehat{\underline Z}_{dn}^{\mathsf T}, \underline Z_{n+1}^*)$, where $\underline Z_{n+1}^*$ is the bootstrap sample for the $1$-step ahead future observation for the Gaussian process $\underline Z_t$. Let $\underline Y_{i,n+1}^* = \widehat F_i^{-1}(\Phi(Z_{i,n+1}^*))$ for $i = 1,\cdots,d$.
		\item Let $\underline R_{n+1}^* = \underline Y_{n+1}^* - \underline{\widehat Y}_{n+1}^*$.
	\end{enumerate}
\item Use step 4 to bootstrap $\underline R_{n+1}^*$ $B$ times. Let $r^{*(b)}_p = \norm{\underline R_{n+1}^{*(b)}}_p$ and $q_{\alpha}^*$ the upper $\alpha$-quantile for $\{r^{*(b)}_p\}_{b=1}^B$. The  $L^p$-norm based $1-\alpha$ prediction region for $\underline Y_{n+1}$ is 
$$\{\underline y\in\mathbb R^d: \norm{\underline y - \widehat{\underline Y}_{n+1}}_p \leq q_{\alpha}^*\}.$$
\end{enumerate}
\end{algorithm}
\begin{remark}(Choice of $L_p$-norm)
	Different $p$ values will affect the shape of the prediction region. Some common choices for $p$ include the $p = 1$, $p = 2$, and $p = \infty$. For example, using $p=2$ will produce a $d-$dimensional ball; while using $p = \infty$ will produce a $d-$dimensional rectangle. 
\end{remark}
\begin{remark}(Prediction region based on studentized root)
Step 5 of algorithm \ref{alg:1} can be augmented to produce predictive region based on studentized predictive root. Let $\widehat{\mathbf V}_n$ be the estimated covariance matrix of $\underline R_{n+1}$, and $\widehat{\mathbf V}_n^*$ the estimated covariance matrix of $\underline R_{n+1}^*$, then the studentized root $\widehat{\mathbf V}_n^{*-1/2}\underline R_{n+1}^*$ can be used to replace 4(c). The corresponding prediction region is then 
$$\{\underline y\in\mathbb R^d: \norm{\widehat{\mathbf V}_n^{-1/2}\left(\underline y - \widehat{\underline Y}_{n+1}\right)}_p \leq {q'_{\alpha}}^{*}\},$$
where ${q'_{\alpha}}^{*}$ is the analog of $q_{\alpha}^*$ for the studentized root.
\end{remark}
\begin{remark}(Limit model-free bootstrap)
	The entries of ${\underline \xi}_{dn}^*$ and ${\underline \xi}_{d}^*$ can also be sampled from $\mathcal N(0,1)$, which will be the limiting distribution for the entries of $\widehat{\underline \xi}_{dn}$. This is called the limit model-free bootstrap in \cite{Politis2015} and has superior performance under certain scenarios.
\end{remark}
\begin{remark}(Generalization to $h$-step ahead prediction region)
	The above bootstrap algorithm can be generalized to replicate the $h$-step ahead predictive root
	$$\underline R_{n+h} = \underline Y_{n+h} - \widehat Y_{n+h},$$
	based on which $h$-step ahead prediction region can be constructed as well. This is useful for the next section, where we propose a new approach for constructing joint prediction bands for univariate time series. 
\end{remark}
\begin{remark}(MFB with fixed predictor)
	Another variant to bootstrap $\underline R_{n+1}$ is to generate 
	$$\underline{\widetilde{R}}_{n+1}^* = \underline Y_{n+1}^* - \widehat{ \underline Y}_{n+1},$$
	for faster execution time, more stability, and (possibly) better performance. We compare the bootstrap performance of both procedures in Section \ref{sec:numerical}.
\end{remark}

Algorithm \ref{alg:1} will work under model assumption \eqref{model:1}. In order to perform bootstrap for the more complicated model \eqref{model:2}, certain changes need to be applied to the above algorithm detailed below:

\begin{algorithm}\label{alg:2} Bootstrap algorithm under model \ref{model:2}:
	
	Replace each $F_i$, $F_i^{-1}$ by $ F_{i,i-1}(\cdot) = \mathbb P(Y_{i,t}\leq \cdot|\underline Y_{1:i-1,t})$ and its inverse.  Possible estimator choices for $F_{i,i-1}(\cdot)$ include the nonparametric conditional CDF estimator, and estimator based on quantile/distributional regression.
\end{algorithm}

\subsection{Generating joint prediction band for univariate time series}\label{sec:multivariate_JPB_for_univariate}
Apart from producing prediction regions for multivariate time series, the above algorithm can also be used under a univariate time series setup. 
For this section, consider $\{Y_t\}_{t=1}^n$ to be realizations from model \ref{model:1} with $d = 1$. 
The problem of interest now is to generate a joint prediction band(JPB) for observations from time $n+1$ up until $n+h$.

First of all, the algorithm described in \cite{WangPolitis2019} can be used to generate  prediction intervals(PI) of level $1-\alpha$ for observations up to $h-$step ahead, $\{\widehat C_{1-\alpha}(j)\}_{j=n+1}^{n+h}$, such that asymptotic validity holds for all the PIs: as $n\rightarrow\infty$,
$$\sup_{n+1 \leq j \leq n+h}\mathbb \lvert \mathbb P\left(Y_{t+j}\in \widehat C_{1-\alpha}(j)| \{Y_t\}_{t=1}^n\right)- (1-\alpha)\rvert\rightarrow 0.$$
 A straightforward solution to adapting towards a JPB $\widehat{\mathbf C}_{n+1:n+h}$ is through the Bonferroni correction: let 
 
 $$\widehat{\mathbf C}_{n+1:n+h}^{(Bon)} = \bigtimes_{j=n+1}^{n+h} \widehat C_{1-\alpha/h}(j)$$
 Then by simple union bounds, for large enough $n$,
 \begin{equation}\label{eq:bonferroni}
 	\mathbb P\left(\underline Y_{n+1:n+h}\in\widehat{\mathbf C}_{n+1:n+h}^{(Bon)}| \{Y_t\}_{t=1}^n\right) \geq 1-\alpha.
 \end{equation}
However, equation \eqref{eq:bonferroni} does not guarantee an exact $1-\alpha$ coverage; it is also well recognized that  Bonferroni correction is a conservative method that produces prediction regions with coverage much larger than nominal levels, particularly so when the prediction intervals $\widehat C_{1-\alpha/h}(j)$ are correlated, which is indeed the case for this time series setting.

A different perspective for generating valid JPBs is through controlling errors arising from multiple testing, such as family-wise errors(FWE) or false discovery rate. For example, \cite{bootstrapJPB} proposed general bootstrap methods for producing JPBs with guaranteed control for {\it k-FWE}: the error associated with false coverage for at least $k$ future observations. 

Here, we offer a new approach for generating JPB with guaranteed coverage level utilizing the model-free bootstrap algorithm described above, through {\it stacking} the univariate time series. With a slight abuse of notation, for this section, let  
$$\underline Y_{t} = (Y_{t-h+1},\cdots,Y_{t})^{\mathsf T},\, t = h,\cdots,n.$$
Then $\{\underline Y_{t}\}_{t=h}^{n}$ are  current observations stacked into vectors of dimension $h$, and the $h$-step ahead data vector $\underline Y_{n+h} = (Y_{n+1},\cdots,Y_{n+h})^{\mathsf T}$ consists of the next $h$ future observations. Since $Y_t$ follows model \ref{model:1} with $d=1$, the stacked version $\underline Y_t$ also satisfies model \ref{model:1} with $d = h$. This facilitates using algorithm \ref{alg:1} to construct a $1-\alpha$ prediction region for $\underline Y_{n+h}$, which is also the JPB for the next $h$ observations of $Y_t$.

\section{Theoretical Results}
In this section, we prove bootstrap prediction region validity of algorithm \ref{alg:1} by using the same technique of \cite{WangPolitis2019}. 
The proof  mostly follows from the proof of Theorem 5.1 in \cite{WangPolitis2019}.  
The main challenge is to show 
$$\norm{\duohat{\boldsymbol\Gamma}_{dn} - {\boldsymbol\Gamma}_{dn} }_{op}\overset{P}{\rightarrow} 0,$$
under certain assumptions. The additional difference under multivariate setting is that ${\boldsymbol\Gamma}_{dn}$ is no longer a Toeplitz matrix, as is the case in \cite{WangPolitis2019}, but rather a block Toeplitz matrix.

We list the following general assumptions:
\begin{assumption}
	\vspace{0.3cm}
	(A1). $\underline Y_t \in \mathbb R^d$ follows model \eqref{model:1}, with each $f_i$ continuously differentiable and strictly monotone.
	\vspace{0.1cm}
	
	(A2). For all $1\leq i\leq d$, the estimator $\widehat F_i$ satisfies a uniform consistency condition with $\mathcal O_p(1/\sqrt{n})$ rate:
	\begin{equation}
		\sup_{y\in\mathbb R} |\widehat F_i(y) - F_i(y)| = \mathcal O_p(1/\sqrt{n}).
	\end{equation}
\vspace{0.1cm}

	(A3).  $\sum_{h=0}^\infty \lvert{\boldsymbol\Gamma_h}\rvert_1 <\infty.$
\vspace{0.1cm}
	
	(A4). $\exists M>0$, such that for all $i, j = 1,\cdots, d$, and all $ |h|<n$,
	
	$$\mathbb E \lvert \sum_{t=1}^{n-|h|} Z_{i,t}Z_{j,t+|h|} - n\boldsymbol\Gamma_{|h|}(i,j)\rvert_1\leq M\sqrt{n}.$$
\vspace{0.1cm}
	
	(A5). $\exists \lambda_0, \epsilon_0 >0, n_0\in\mathbb N$, such that the eigenvalues of $\boldsymbol \Gamma_{dn}$ are uniformly bounded and bounded away from $0$ for all $n>n_0$, i.e.
	$$\lambda_0\geq \lambda_{max} (\boldsymbol \Gamma_{dn})\geq\lambda_{min} (\boldsymbol \Gamma_{dn}) \geq \epsilon_0 >0.$$
\vspace{0.1cm}

(A6). Both $l$ and $c$ diverge to infinity as $n\rightarrow\infty$, such that $1/l + l/\sqrt{n} = o(1)$, $lce^{c^2/2} = o(\sqrt{n})$, and $lc^{1/2}e^{-c^2/4} = o(1)$.
\end{assumption}

\begin{prop} \label{prop:uniform_consistency}
	Under assumptions (A1) and (A2), $\underline Z_t$ is a multivariate Gaussian process, and 
	$$\sup_{1\leq t\leq n}|\widehat U_{i,t} - U_{i,t}| = \mathcal O_p(1/\sqrt{n}).$$
\end{prop}

\begin{lemma}\label{lm:matrix_convergence}
	Under assumptions (A1)-(A5), as $n\rightarrow\infty$,  $\duohat{\boldsymbol\Gamma}_{dn}$ is positive definite in probability; also, both
	 $\norm{\duohat{\boldsymbol\Gamma}_{dn} - {\boldsymbol\Gamma}_{dn} }_{op}$ and 
	$\norm{\duohat{\boldsymbol\Gamma}_{dn}^{-1} - {\boldsymbol\Gamma}_{dn}^{-1} }_{op} $ converge to $0$ in probability.
\end{lemma}
\begin{proof}
	First of all, since $\duohat{\boldsymbol\Gamma}_{dn} - {\boldsymbol\Gamma}_{dn}$ is symmetric, 
	\begin{equation*}\label{eq:ineq1}
		\begin{split}
	\norm{\duohat{\boldsymbol\Gamma}_{dn} - {\boldsymbol\Gamma}_{dn} }_{op} 
&\leq
	 \sqrt{\norm{\duohat{\boldsymbol\Gamma}_{dn} - {\boldsymbol\Gamma}_{dn} }_{1}\norm{\duohat{\boldsymbol\Gamma}_{dn} - {\boldsymbol\Gamma}_{dn} }_{\infty}}\\
&=
 \norm{\duohat{\boldsymbol\Gamma}_{dn} - {\boldsymbol\Gamma}_{dn} }_{\infty} \\
 & = \max_{1\leq i\leq dn} \sum_{j=1}^{dn} \lvert(\duohat{\boldsymbol\Gamma}_{dn} - {\boldsymbol\Gamma}_{dn})(i,j)\rvert.
		\end{split}
	\end{equation*}
By block Toeplitz property,
\begin{equation}\label{eq:ineq2}
	\begin{split}
		\max_{1\leq i\leq dn} \sum_{j=1}^{dn} \lvert(\duohat{\boldsymbol\Gamma}_{dn} - {\boldsymbol\Gamma}_{dn})(i,j)\rvert
		&\leq 
		\max_{0\leq k \leq n-1} \sum_{i = kd+1}^{(k+1)d}\sum_{j=1}^{dn}\lvert(\duohat{\boldsymbol\Gamma}_{dn} - {\boldsymbol\Gamma}_{dn})(i,j)\rvert\\
		&\leq
		\sum_{h = 1-n}^{n-1} \sum_{i=1}^d\sum_{j=1}^d \lvert(\duohat{\boldsymbol\Gamma}_h - \boldsymbol\Gamma_h)(i,j)\rvert.
	\end{split}
\end{equation}
Notice that $\sum_{i=1}^d\sum_{j=1}^d \lvert(\duohat{\boldsymbol\Gamma}_h - \boldsymbol\Gamma_h)(i,j)\rvert = |\duohat{\boldsymbol\Gamma}_h - \boldsymbol\Gamma_h|_1$ is the entry-wise $l_1$ norm, thus the RHS of equation \eqref{eq:ineq2} equals
\begin{equation*}
	\begin{split}
		\sum_{h = 1-n}^{n-1} \lvert\duohat{\boldsymbol\Gamma}_h - \boldsymbol\Gamma_h\rvert_1 
		&=
		\sum_{h = 1-n}^{n-1} \lvert\duohat{\boldsymbol\Gamma}_h - \widehat{\boldsymbol\Gamma}_h + \widehat{\boldsymbol\Gamma}_h - \boldsymbol\Gamma_h\rvert_1\\
		&\leq
		\sum_{h = 1-n}^{n-1}\lvert\duohat{\boldsymbol\Gamma}_h - \widehat{\boldsymbol\Gamma}_h\rvert_1 +  \sum_{h = 1-n}^{n-1}\lvert\widehat{\boldsymbol\Gamma}_h - \boldsymbol\Gamma_h\rvert_1.
	\end{split}
\end{equation*}

 The second sum was proved to converge to $0$ in probability by Theorem 2.1, \cite{Jentsch} under (A3), (A4) and (A6). We only need to show the first sum $\sum_{h=1-n}^{n-1}\lvert\duohat{\boldsymbol\Gamma}_h - \widehat{\boldsymbol\Gamma}_h\rvert_1 $  converges to $0$ as well.
 \begin{equation}
 	\begin{split}
\sum_{h=1-n}^{n-1}\lvert\duohat{\boldsymbol\Gamma}_h - \widehat{\boldsymbol\Gamma}_h\rvert_1 
&=
\sum_{1\leq i,j\leq d} \sum_{h=1-n}^{n-1} \lvert\duohat{\boldsymbol\Gamma}_h(i,j) -\widehat{ \boldsymbol\Gamma}_h(i,j)\rvert\\
&= \sum_{1\leq i,j\leq d} \sum_{h=1-n}^{n-1}
\lvert\frac{1}{n} \sum_{t=1}^{n-|h|}(\tildehat Z_{i,t}
\tildehat Z_{j,t+h} -  Z_{i,t}Z_{j,t+h})\rvert,
 	\end{split}
 \end{equation}
As we assume the dimension $d$ is fixed, we only need to show the second sum
$ \sum_{h=1-n}^{n-1}
\lvert\frac{1}{n} \sum_{t=1}^{n-|h|}(\tildehat Z_{i,t}
\tildehat Z_{j,t+h} -  Z_{i,t}Z_{j,t+h})\rvert$
converges to $0$ for all pairs of $(i,j)$. Under the result of proposition \ref{prop:uniform_consistency} and (A6), \cite{WangPolitis2019} shows the sum does converge to 0 in probability. Thus $\sum_{h=1-n}^{n-1}\lvert\duohat{\boldsymbol\Gamma}_h - \widehat{\boldsymbol\Gamma}_h\rvert_1 \overset{P}{\rightarrow} 0$, and $\norm{\duohat{\boldsymbol\Gamma}_{dn} - {\boldsymbol\Gamma}_{dn} }_{op}\overset{P}{\rightarrow} 0$.
\end{proof}
\begin{remark}
	For the next theorem, we need a slightly stronger result of Lemma \ref{lm:matrix_convergence}, which is
	\begin{equation}\label{eq:conv_with_rate}
	\frac{1}{(\log n)^2}\norm{\duohat{\boldsymbol\Gamma}_{dn} - {\boldsymbol\Gamma}_{dn} }_{op}\overset{P}{\rightarrow} 0.
	\end{equation}
	This guarantees $\norm{\duohat{\boldsymbol\Gamma}_{dn}^{1/2} - {\boldsymbol\Gamma}_{dn}^{1/2} }_{op}\overset{P}{\rightarrow} 0$ as well as $\norm{\duohat{\boldsymbol\Gamma}_{dn}^{-1/2} - {\boldsymbol\Gamma}_{dn}^{-1/2} }_{op}\overset{P}{\rightarrow} 0$. Note that \eqref{eq:conv_with_rate} is achievable if we choose appropriate rates for $l$ and $c$.
\end{remark}
 
\begin{theorem}\label{thm:asymptotic_validity}
	Under (A1) - (A6) such that \eqref{eq:conv_with_rate} also holds,
	 then the predictive distribution $F_{n+1|n}$ is continuous, and 
	\begin{equation}\label{eq:consistency_pred_distribution}
		\sup_{\underline y\in\mathbb R^d} |\widehat F_{n+1|n}(\underline y) - F_{n+1|n}(\underline y)|\overset{P}{\rightarrow} 0.
	\end{equation}
Also, assume that $\mathbb E_{\underline Y\sim F_{n+1|n}} \norm{\underline Y}_p <\infty$, and that under equation \eqref{eq:consistency_pred_distribution}, $\widehat{\underline Y}_{n+1} \rightarrow \mathcal P(\underline Y_{n+1}|\mathcal F_n)$ in probability. Then the prediction region generated by the MFB algorithm \ref{alg:1} is asymptotically valid.
\end{theorem}
The proof of Theorem \ref{thm:asymptotic_validity} proceeds the same  as that of Theorem 5.1 of \cite{WangPolitis2019}.

\section{Numerical Results}\label{sec:numerical}
\subsection{Synthetic data experiment}
\textbf{Experiment setup.} Our main focus in this section is to numerically demonstrate the coverage performance of the multivariate MFB. Consider the following nonlinear time series model with dimension $d = 2$:
$$
Y_{i,t} = f_i(W_{i,t}),
$$
where $f_1(x) = f_2(x) = sgn(x) \sqrt{|x|}$, and $\underline W_t$ is a VAR(1) process with Gaussian innovations:
$$\underline W_t = \boldsymbol A\underline W_{t-1} + \underline \epsilon_t,$$
where we set $\boldsymbol A = \begin{pmatrix}
	0.5 & 0.2 \\
	0.2 & 0.6
\end{pmatrix}$, $\underline \epsilon_t \overset{i.i.d.}{\sim} \mathcal N(\mathbf 0, \boldsymbol B)$ with $\boldsymbol{B} = \begin{pmatrix}
2 & 0.5\\
0.5 & 2
\end{pmatrix}$. 
Since the $f_i$s are continuous monotone functions and $\underline W_t$ is a Gaussian process, $\underline Y_t$ satisfies model \ref{model:1}. 

Let $\mathbf Y_n = \{\underline Y_t\}_{t=1}^n$ be an arbitrary sample path drawn from above. 
By applying algorithm \ref{alg:1} we then generate a two sided prediction region: $\widehat {\mathbf C}_{1-\alpha}(\mathbf Y_n)$. 
Let $\{\underline Y_{n+1}^{(m)}\}_{m=1}^M$ be a set of future observations sampled directly from the above model conditioning on $\mathbf Y_n$, namely $\underline Y_{n+1}^{(m)}\overset{i.i.d.}{\sim} F_{n+1|n}$.  The coverage probability of $\widehat {\mathbf C}_{1-\alpha}(\mathbf Y_n)$ 
is estimated by the empirical coverage rate for $\{\underline Y_{n+1}^{(m)}\}_{m=1}^M$, i.e.,

$$\widehat{CVR}(\mathbf Y_n) = \frac{1}{M}\sum_{m=1}^M  I\left(\underline Y_{n+1}^{(m)}\in\widehat {\mathbf C}_{1-\alpha}(\mathbf Y_n)\right).$$

To achieve stable estimations for the coverage probability, we generate multiple realizations of $\mathbf Y_n$ and calculate the average of the $CVR$ estimations as our  metric. 

\textbf{Parameter selection.} The following parameters need to be properly selected in algorithm \ref{alg:1}: the bandwidths $\{b_i\}_{i=1}^d$ for the nonparametric CDFs $\{\widehat F_i\}_{i=1}^d$; the banding parameter $l$ in $\duohat{\boldsymbol \Gamma}_{dn}$; the choice of optimal predictor as well as the $L_p$ norm when evaluating the root $\underline R_{n+1}$.

While there are both sample-based rules as well as cross-validation for selecting $b_i$ and $l$, see e.g, \cite{Das2017} and \cite{mcmurrypolitis2010}, 
 in this experiment we take another approach. We use a fixed set of parameters for the bandwidths $\{b_i\}_{i=1}^d$, and evaluate the performance of the MFB algorithm with various combinations of the remaining parameters. This way, we can see the effect of each parameter on the performance of MFB.
 The coverage metrics are obtained for a range of sample sizes $n\in\{100, 200, 300, 400, 500\}$. We set $\alpha = 0.05$ which means the nominal coverage should be $95\%$.
 Results of our simulations are plotted below in Figure \ref{fig:cvr_n=100} to \ref{fig:cvr_n=500} . 

As expected, as $n$ increases, the coverage probability converges to the nominal $1-\alpha$, showing asymptotic validity. For this particular example, the choice of $L^p$ norm and predictor type has effect on the relative advantage between MFB with resampled/fixed predictor. 
Note that if we compare the best performance of the two algorithms according to closeness of coverage probability with respect to the nominal level, both seem to work very well under certain parameter combinations.

 \begin{figure}[h]
	\centering
 \includegraphics[scale = 0.8]{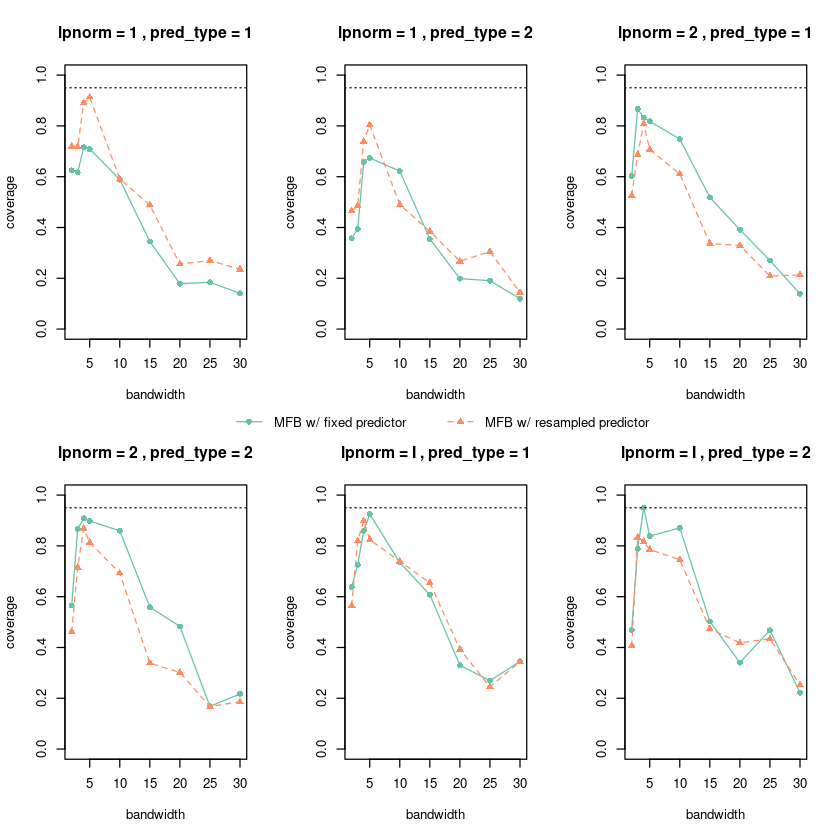}
	\caption{CVR plots with $n=100$.}
	\label{fig:cvr_n=100}
 \end{figure}
 
  \begin{figure}[h]
 	\centering
 	\includegraphics[scale = 0.8]{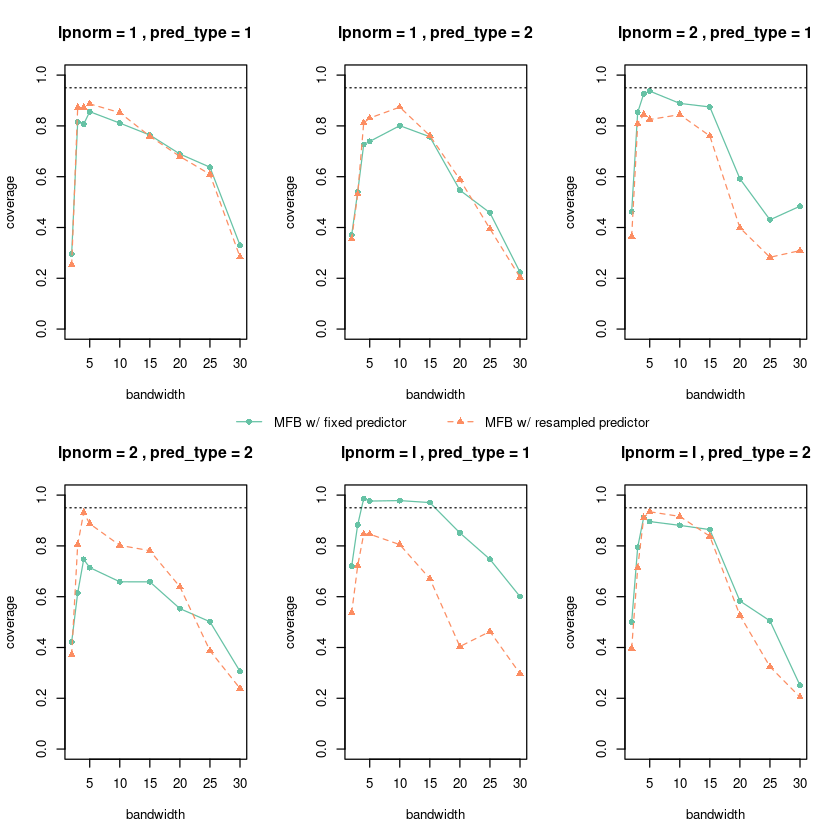}
 	\caption{CVR plots with $n=200$.}
 	\label{fig:cvr_n=200}
 \end{figure} \begin{figure}[h]
 \centering
 \includegraphics[scale = 0.8]{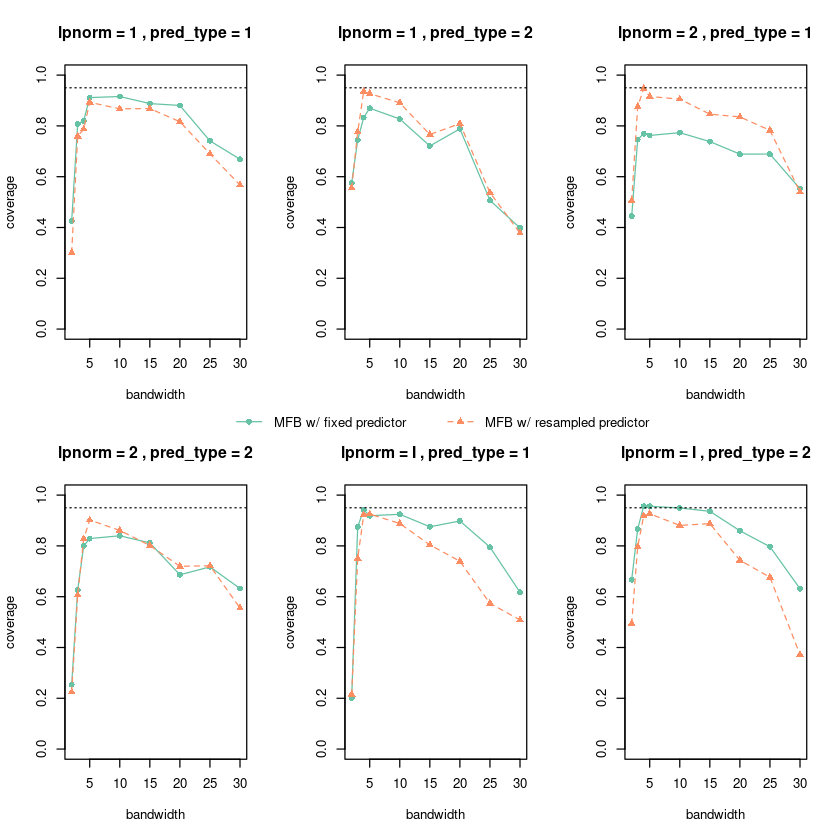}
 \caption{CVR plots with $n=300$.}
 \label{fig:cvr_n=300}
\end{figure} \begin{figure}[h]
\centering
\includegraphics[scale = 0.8]{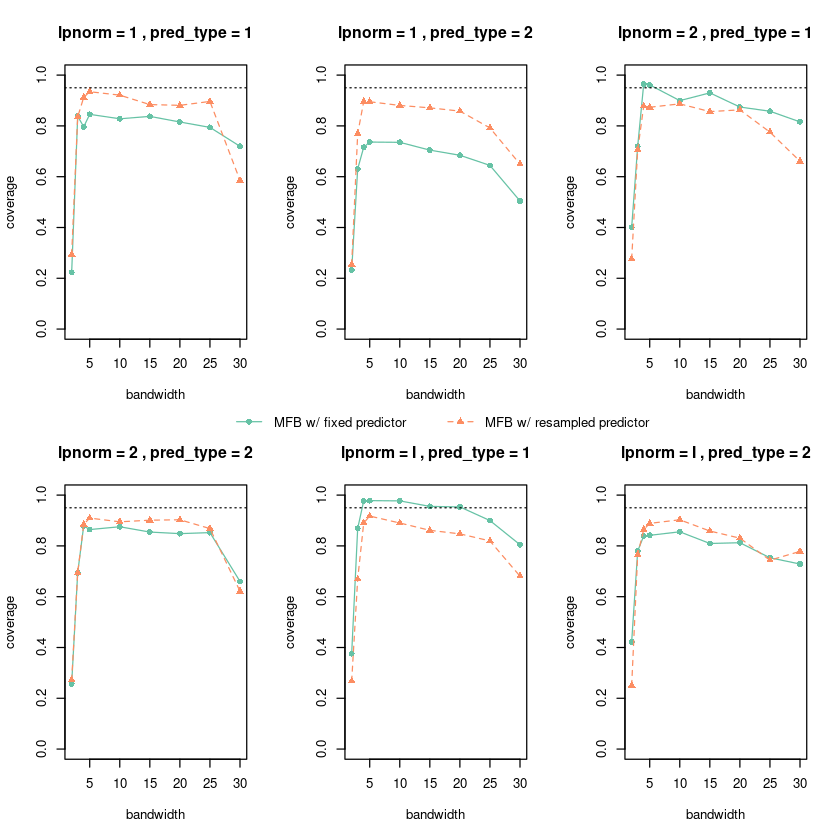}
\caption{CVR plots with $n=400$.}
\label{fig:cvr_n=400}
\end{figure} \begin{figure}[h]
\centering
\includegraphics[scale = 0.8]{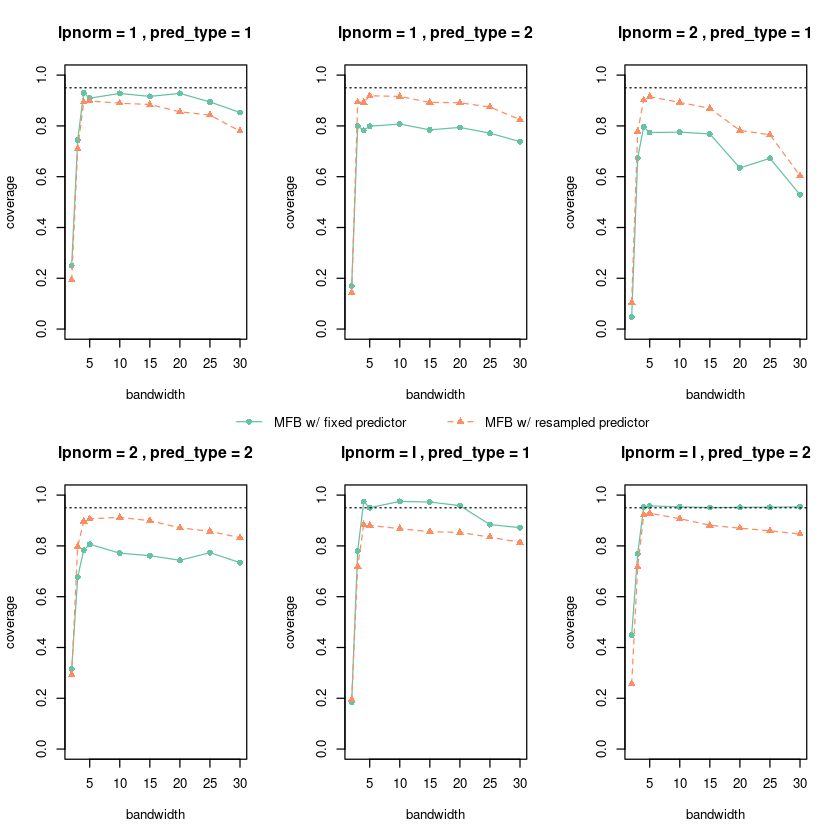}
\caption{CVR plots with $n=500$.}
\label{fig:cvr_n=500}
\end{figure}

 \subsection{Joint prediction band for  heavy-tailed financial returns}
 A lot of real world time series data exhibit heavy-tailed phenomena, particularly  exemplified under the realm of financial time series, wherein a common characteristic is the so-called leptokurtosis, i.e.,
 $$Kurt = \frac{\mu_4}{\sigma^4} >3,$$
 where $\mu_4 = \mathbb E(Y_t - \mu)^4$ is the centered fourth moment of the univariate time series $Y_t$, and $\sigma = \sqrt{\mathbb E(Y_t - \mu)^2}$ is the standard deviation. 
 Under a leptokurtic situation, the marginal distribution of the time series has heavier tail than the normal distribution. 
 Accurate modeling for heavy-tailedness 
is of great importance in the financial market.  
For example, Value-at-Risk(VaR) is a fundamental  metric commonly used to measure the risk associated with an asset, defined as 
$$VaR(\alpha) = \inf\{y\in\mathbb R: F(y)\geq \alpha\}.$$
Clearly, VaR captures the (left) heavy-tailedness of the CDF $F$. 

A common approach for modeling heavy-tailedness of financial data is through a Gaussian mixture model: Let $Y_t = \sigma_t z_t$, where $\sigma_t$ is random and $\mathcal F_{t-1}$-measurable; $z_t\sim \mathcal N(0,1)$. Then $Y_t$ is distributed according to a Gaussian mixture, thus leptokurtic. The famous ARCH/GARCH model, and the more recent NoVaS transformation of \cite{novas_2007}  are both paragons of this approach.

A different approach is to model heavy-tailedness via nonlinear transformation models. In particular, 
$$Y_t = f(W_t)$$
of model \ref{model:1} is a classic nonlinear model studied by many, both from a theoretical perspective (cf. \cite{nonstablegaussian}, \cite{Breuer1983}) and also an applied perspective(cf. \cite{Hull1998}). 
As previously discussed in Section \ref{sec:multivariate_JPB_for_univariate}, the multivariate MFB algorithm can be used to generate JPB under this setup. The JPB will be useful to describe probable region of future paths and thus provide more information for  trading. 
In this section, we present some numerical results on the empirical coverage performance of the MFB-based JPB for daily stock returns.
We also compare them with a benchmark method, which is JPB based on a garch(1,1)-bootstrap -- see \cite{econometrics7030034} for details. 

The details of our experiments goes as follows. We pick the following stocks: AAPL, AMZN, TSLA, GME as candidate datasets, where we gathered daily stock returns from 2018-01-01 to 2021-08-31, denoted as  $\{Y_t\}_{t=1}^n$ with $n = 922$. 
We then sequentially generate a collection of $(past, future)$ pairs:
\begin{equation}\label{exp:stock_data}
\left\{ \left(  \{Y_{i}\}_{i=t-n_0+1}^t, \{Y_{j}\}_{j=t+1}^{t+h}\right): t = n_0 + kh, 0\leq k\leq \lfloor\frac{n-n_0}{h}\rfloor\right\}.
\end{equation}
The parameter $n_0$ represents the number of days used to backtrack past data for prediction purpose; and $h$ represents the dimension of future data we try to predict.
By using the data generation scheme in \eqref{exp:stock_data}, the future observations are non-overlapping and are better suited for calculating the empirical coverage. 
To simplify notations, we let 
$\underline X_{k,1} = (Y_{kh+1},\cdots,Y_{kh+n_0})^{\mathsf T}$ and $\underline X_{k,2} = (Y_{n_0 + kh + 1},\cdots, Y_{n_0 + (k+1)h})^{\mathsf T}.$
The JPB calculated from the data $\underline X_{k,1}$ is denoted by $\widehat {\mathbf C}(\underline  X_{k,1} )$.
Finally, the empirical coverage rate(ECVR) based on data $Y_t$ and parameters $n_0$, $h$
is calculated by
$$ECVR(\{Y_t\}_{t=1}^n, n_0, h) = \frac{1}{\lfloor\frac{n-n_0}{h}\rfloor + 1} \sum_{k=0}^{\lfloor\frac{n-n_0}{h}\rfloor} \underline  X_{k,2} \in \widehat {\mathbf C}(\underline  X_{k,1} ).$$

We plot the ECVR against a range of $n_0$ values for different stocks and $h$. In the MFB algorithm, we choose to use the $L^2$-optimal predictor and the $L^1$ norm for root evaluation; the bandwidth $b = 0.01$ and banding parameter $l = 0.4$ are determined via cross validation. The results of our experiments are presented in Figure \ref{fig:ecvr_AAPL} - \ref{fig:ecvr_TSLA}.

 \begin{figure}[h]
	\centering
	\includegraphics[scale = 0.25]{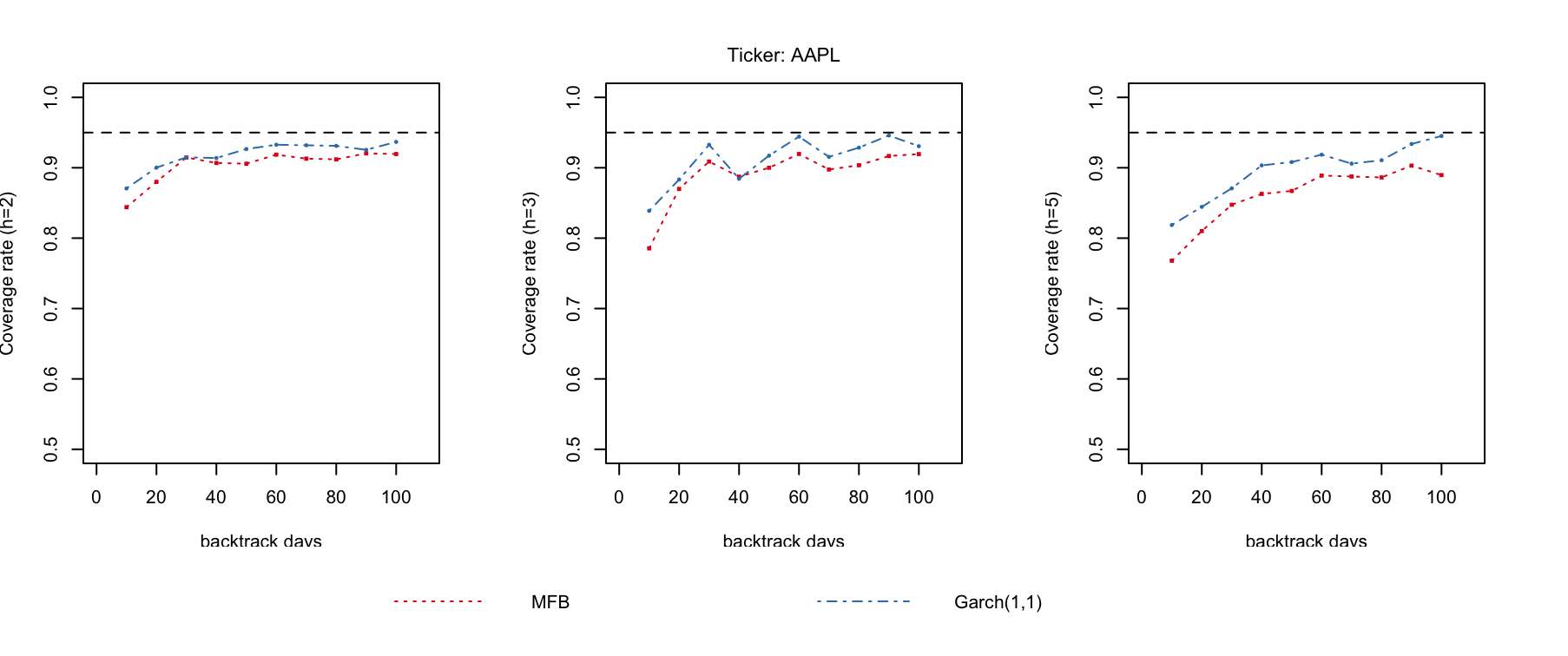}
	\caption{ECVR plots for AAPL.}
	\label{fig:ecvr_AAPL}
\end{figure} 

\begin{figure}[h]
	\centering
	\includegraphics[scale = 0.32]{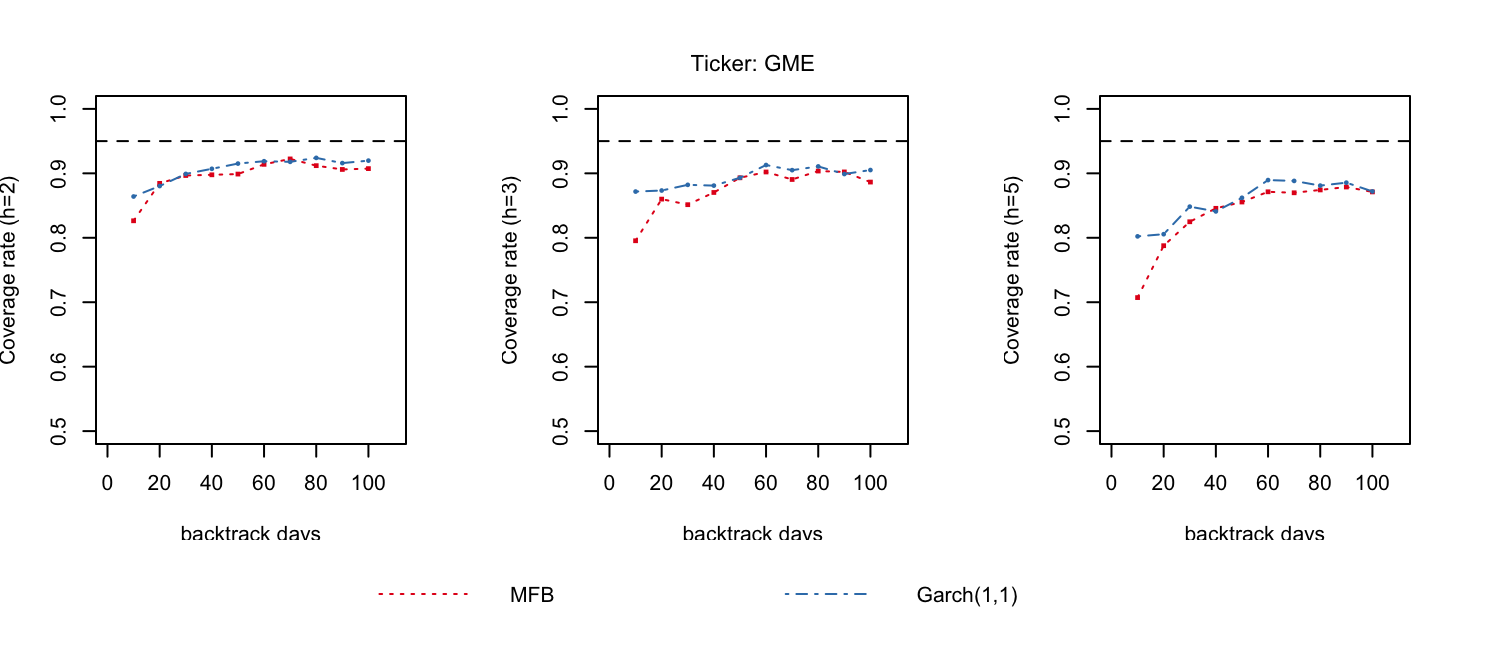}
	\caption{ECVR plots for GME}
	\label{fig:ecvr_GME}
\end{figure} 
\begin{figure}[h]
\centering
\includegraphics[scale = 0.25]{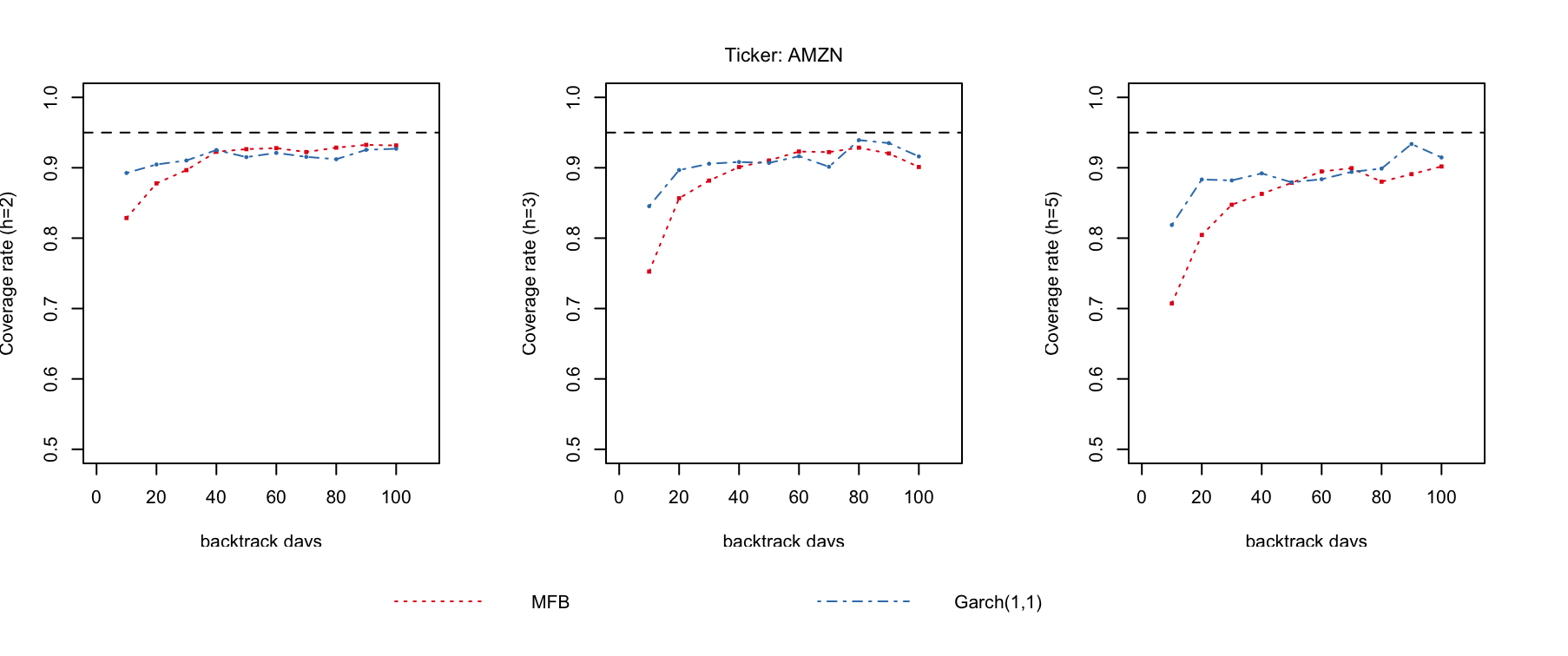}
\caption{ECVR plots for AMZN.}
\label{fig:ecvr_AMZN}
\end{figure} 
\begin{figure}[h]
\centering
\includegraphics[scale = 0.25]{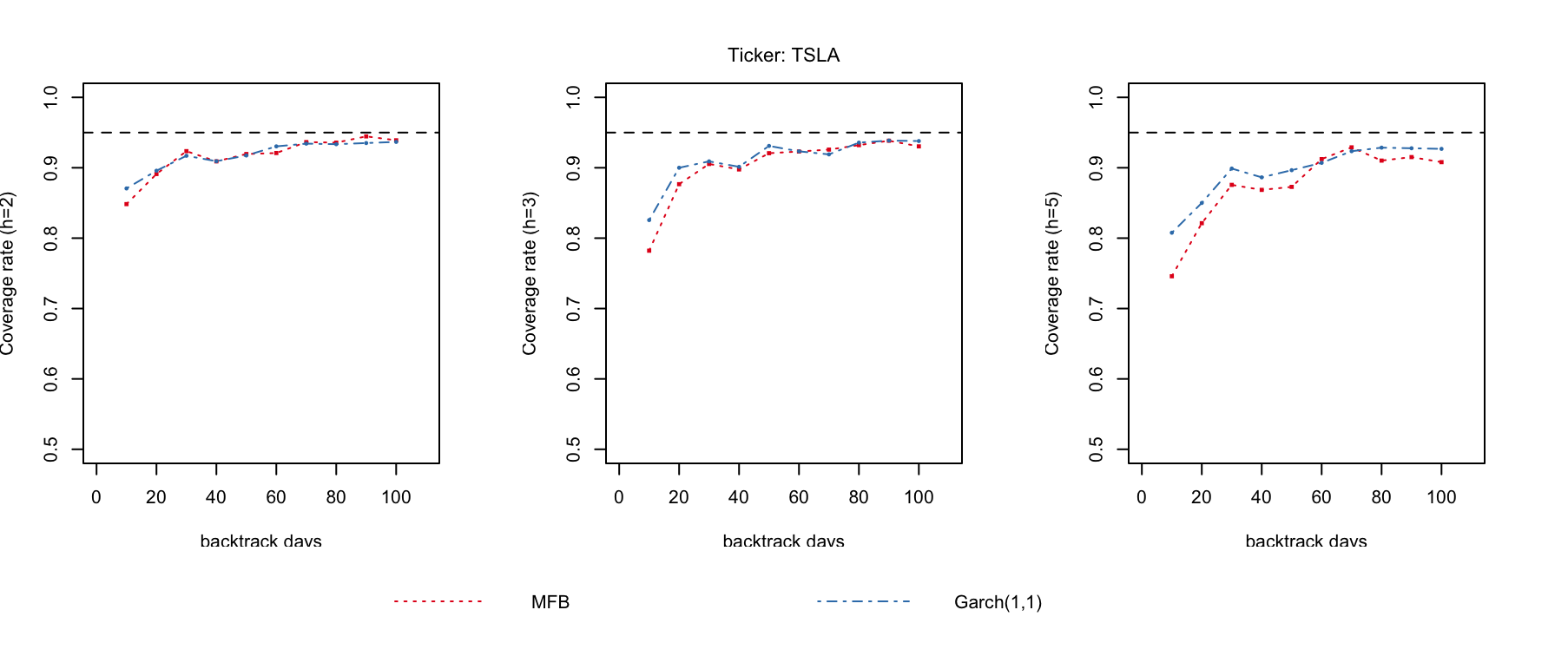}
\caption{ECVR plots for TSLA.}
\label{fig:ecvr_TSLA}
\end{figure} 

We can see that as the number of backtracking days increases, the empirical coverage also grows towards $95\%$ because of asymptotic validity. However, both methods still have under coverage issue, partially due to the fact that the stationarity assumption does not perfectly fit the stock market.
We also observe that for $h=2$ and $3$, the performance of the two methods are on par with each other. However, for $h=5$ the garch-based bootstrap has superior performance comparing with the MFB. This is due to the more efficient parameter searching scheme in garch fitting, whereas for the MFB, the process of parameter searching mostly relies on cross validation, which can be inefficient and suboptimal.

\bibliography{reference}
\end{document}